\documentclass[12pt]{article}

\usepackage{amsmath}
\usepackage{amsthm}
\usepackage{amssymb}
\usepackage{rotating}
\usepackage{array}
\usepackage{enumitem}
\usepackage{mathtools}
\usepackage{setspace}
\usepackage{fullpage}
\usepackage{hyperref}
\usepackage{booktabs}
\usepackage{algorithm}
\usepackage[noend]{algpseudocode}
\usepackage{xspace}
\usepackage{authblk}

\newtheorem{definition}{Definition}
\newtheorem{theorem}{Theorem}
\newtheorem{lemma}{Lemma}

\newcommand{\R}{\mathbb{R}}
\newcommand{\N}{\mathbb{N}}
\newcommand{\G}{\mathbb{G}}

\newcommand{\BAn}{Barab\'{a}si-Albert\xspace}
\newcommand{\ERn}{Erd\H{o}s-R\'{e}nyi\xspace}
\newcommand{\WSn}{Watts-Strogatz\xspace}

\newcommand{\cdeg}{c_{deg}}
\newcommand{\cclo}{c_{clo}}
\newcommand{\cbet}{c_{bet}}
\newcommand{\cged}{c_{ged}}

\newcommand{\HH}{H}
\newcommand{\FR}{\hat{R}}
\newcommand{\RR}{R^*}
\newcommand{\thr}{\theta}

\title{Network Members Can Hide from Group Centrality Measures}

\author[a,b]{Marcin Waniek}
\author[a]{Talal Rahwan}

\affil[a]{Computer Science, New York University Abu Dhabi, Abu Dhabi, UAE}
\affil[b]{Institute of Informatics, University of Warsaw, Warsaw, Poland}

\date{}

\begin{document}

\maketitle

\begin{abstract}
Group centrality measures are a generalization of standard centrality, designed to quantify the importance of not just a single node (as is the case with standard measures) but rather that of a group of nodes. Some nodes may have an incentive to evade such measures, i.e., to hide their actual importance, in order to conceal their true role in the network. A number of studies have been proposed in the literature to understand how nodes can rewire the network in order to evade standard centrality, but no study has focused on group centrality to date. We close this gap by analyzing four group centrality measures: degree, closeness, betweenness, and GED-walk. We show that an optimal way to rewire the network can be computed efficiently given the former measure, but the problem is NP-complete given closeness and betweenness. Moreover, we empirically evaluate a number of hiding strategies, and show that an optimal way to hide from degree group centrality is also effective in practice against the other measures. Altogether, our results suggest that it is possible to hide from group centrality measures based solely on the local information available to the group members about the network topology.
\end{abstract}

\section{Introduction}
\label{sec:introduction}

As networks continue to permeate many aspects of our lives, the methods that analyze network structures grow in importance. One such method involves using group centrality measures, which quantify the importance of a group of nodes in a network~\cite{everett1999centrality}. These can be thought of as a generalization of standard centrality measures~\cite{bonacich1987power} which evaluate the relative importance of just a single node. It is worth noting that simply aggregating the standard centrality scores of a group of nodes misses some crucial aspects that are captured by the group centrality measures~~\cite{everett2005extending}. Group centrality was used in a wide variety of applications, from the analysis of the citation network of journals~\cite{ni2011degree}, through selecting the facility locations~\cite{fushimi2020facility}, to the deployment strategy of roadside units~\cite{kchiche2009access}.

There are cases where a group of nodes might prefer to avoid getting their importance recognized by an outside observer. Imagine, for example, a cabal of leaders of a criminal or terrorist organization, who would prefer to conceal their significance from law enforcement agencies~\cite{lindelauf2013cooperative,gialampoukidis2016key}. Similarly, consider a group of banks or other institutions colluding in a financial transactions network, and wishing to obfuscate their influence over the system~\cite{chan2018systemic}. Finally, imagine a situation of a security designer wishing to obscure the key components of a system from a potential attacker~\cite{maccari2016computation}. In any of these situations, algorithmic ways of misguiding group centrality analysis might be of use.

A number of works in the literature analyzes the problem of hiding a single node from standard centrality measures~\cite{waniek2018hiding,was2020manipulability,waniek2021strategic}. Similarly, several works considered the task of hiding groups of nodes from standard centrality measures aggregated by taking a maximal centrality value of all members of the group~\cite{waniek2017construction,dey2019covert}. However, the problem of hiding nodes from group centrality measures have not been considered to date.

We close this gap in the literature by providing the first rigorous analysis of the problem of hiding from group centrality measures. To this end, we assume the perspective of a group of \textit{evaders}, i.e., network nodes wishing to evade the analysis performed using group centrality measures. The evaders aim to decrease their group centrality by removing some of their local network connections (we show that it is impossible to decrease their centrality by adding new connections). While designing the problem formulation, we make sure to make it flexible to domain-specific considerations, as we allow to specify which network connections can be removed by the evaders. As a result, the evaders can focus their attention on  relations that can be obscured (e.g., communications using satellite phones in a covert network, or secret meetings between representatives of fraudulent financial institutions), while accepting that some connections are impossible to deny (e.g., blood relations, or transactions on public record).

Having define the problem, we analyze it both theoretically, from the point of view of computational complexity, as well as empirically, using simulations on randomly-generated and real-life networks. We investigate the complexity of finding an optimal way of hiding (i.e., a set of edges the removal of which results in a maximal group centrality drop), and show that it is achievable in polynomial time when hiding from degree group centrality, but the problem becomes intractable for closeness and betweenness centrality. We also show empirically that the optimal strategy against group degree centrality is also effective in practice against the other centralities, including GED-walk. Finally, we investigate what aspects of the network structure itself are crucial to the effectiveness of hiding, showing that its performance is best in small and sparse networks.

The remainder of this article is organized as follows. First, we discuss the relevant works from the literature (Section~\ref{sec:related-work}). Second, we present the preliminaries and the necessary notation (Section~\ref{sec:preliminaries}). Third, we formally state the problem of hiding from group centrality measures and analyze its computational complexity (Section~\ref{sec:problem-statement}). Fourth, we empirically evaluate a number of hiding strategies (Section~\ref{sec:experimental-analysis}). Conclusions follow (Section~\ref{sec:conclusions}).

\section{Related Work}
\label{sec:related-work}

Our work is closely related to the literature on group centrality measures. Everett and Borgatti~\cite{everett1999centrality} define the group variants of the most popular standard centrality measures, i.e., methods designed to measure the importance of a single node. Part of the literature is concerned with the problem of computing group centrality in an efficient way, as they are often more complex than single-node centrality~\cite{puzis2007fast,brandes2008variants}. Significant attention in the literature is devoted to the problem of identifying the most important group according to some centrality measure~\cite{dolev2009incremental,chen2016efficient,veremyev2017finding,Angriman2020Group}. Notice that in our work we assume that the composition of the group of evaders is predefined.

Our work is also related to the literature on hiding from single-node centrality measures in social networks~\cite{was2020manipulability,waniek2021strategic,waniek2023hiding}. Some of these works analyzed the problem of hiding a group of nodes from single-node centrality measures~\cite{waniek2017construction,dey2019covert,waniek2021members}, where the centrality of the group is taken as the maximum over the centralities of its members. Hiding from single-node centrality measures was also considered in alternative network structures, such as multilayer~\cite{waniek2020hiding} and temporal networks~\cite{waniek2022hiding}. The main difference between our work and these mentioned above is that we consider hiding from group centrality measures. Also noteworthy is a growing literature on hiding groups of nodes from community detection algorithms~\cite{waniek2018hiding,chen2019ga,li2020adversarial}. However, the problem considered in this literature is of a very different nature than our work, as community detection algorithms allow to uncover the fact of existence of a previously unknown group, while the group centrality measures allow to quantify the importance of a known group.

A related body of literature concerns itself with hiding from other types of social network analysis tools. These types of evasion techniques are often motivated by the need of privacy protection by the network members~\cite{jin2020adversarial}. Social media users who do not wish some of their undisclosed relationship to be uncovered, might be interested in heuristic solutions designed to mislead link prediction algorithms~\cite{yu2018target,waniek2019hide,zhou2019attacking}. Others might want to counter the analysis performed using node similarity measures~\cite{dey2020manipulating}. Yet another class of techniques allow to hide the identity of the source of network diffusion from the source detection algorithms~\cite{waniek2021social}.

The work that is probably closest in concept to ours is Medya et al.~\cite{medya2018group}. The authors analyze the problem of maximizing the centrality of a group of nodes via network modifications. There is a number of differences between their approach and ours. First, their goal is to increase the group centrality value, whereas we focus on decreasing it. Second, they achieve their goal by adding edges to the network, whereas we consider the problem of removing the edges. Finally, they consider only one centrality measure, namely the coverage centrality (a variant of betweennness centrality), whereas we consider four centrality measures, namely degree, closeness, betweenness, and GED-walk.

\section{Preliminaries}
\label{sec:preliminaries}

In this section we present the basic network notation and concepts used throughout this work. Table~\ref{tab:notation} presents the summary of the notation.

\begin{table}[t]
\centering
\begin{tabular}{c l}
\toprule
Symbol & Meaning \\
\midrule
$V$ & The set of nodes \\
$E$ & The set of edges \\
$N_G(v)$ & The set of neighbors of $v$ in $G$ \\
$N_G(S)$ & The set of neighbors of $S \subseteq V$ in $G$ \\
$d_G(v,w)$ & The distance between $v$ and $w$ in $G$ \\
$\cdeg$ & The degree group centrality \\
$\cclo$ & The closeness group centrality \\
$\cbet$ & The betweenness group centrality \\
$\cged$ & The GED-walk group centrality \\
$\HH$ & The group of evaders \\
$\thr$ & The safety threshold \\
$\FR$ & The set of edges allowed to be removed \\
$b$ & The hiding budget \\
\bottomrule
\end{tabular}
\caption{The summary of notation used in the article.}
\label{tab:notation}
\end{table}

\subsection{Network Notation}

Let $G = (V,E) \in \G$ denote a network with the set of nodes $V$ and the set of edges $E \subseteq V \times V$. In this work we consider undirected networks without self-loops, i.e., we do not distinguish between edges $(v,w)$ and $(w,v)$, and we assume that $\forall_{v} (v,v) \notin E$. When removing edges $R \subseteq E$ from a network $G=(V,E)$, let $G \setminus R$ denote $(V, E \setminus R)$. Similarly, when removing nodes $S \subseteq V$ from a network $G=(V,E)$, let $G \setminus S$ denote $(V \setminus S, E \setminus \{(v,w) \in E : v \in S \lor w \in S \})$.

Let $N_G(v)$ denote the set of neighbors of $v$ in $G$, i.e., $N_G(v) = \{ w \in V : (v,w) \in E \}$. We also extend this notion to the set of neighbors of a group $N(S)$, i.e., $N_G(S) = \{ w \in V \setminus S : \exists_{v \in S} (v,w) \in E \}$. Let $d_G(v,w)$ denote the distance between $v$ and $w$, i.e., the number of edges in a shortest path between $v$ and $w$ in $G$. We assume that if there is no path in $G$ between $v$ and $w$, e.g., because $G$ is not connected, then $d_G(v,w)=\infty$. If the network is clear from context we omit it from the notation, e.g., by writing $N(S)$ instead of $N_G(S)$.

\subsection{Group Centrality}

A group centrality measure quantifies the importance of a group of nodes in a given network~\cite{everett1999centrality}. Formally, it is a function $c : \G \times 2^V \rightarrow \R$, where $c(G,S)$ is the centrality of group $S$ in network $G$. In this work we consider the following group centrality measures:

\begin{itemize}
\item \textbf{Degree group centrality}~\cite{everett1999centrality} is proportional to the number of non-group nodes adjacent to at least one member of the group:
\[
\cdeg(G,S) = \frac{|N_G(S)|}{|V|-|S|}.
\]
\item \textbf{Closeness group centrality}~\cite{everett1999centrality} is proportional to the inverse sum of distances from the group to all non-group nodes:
\[
\cclo(G,S) = \frac{|V|-|S|}{\sum_{v \in V \setminus S} \min_{w \in S} d(w,v)}.
\]
In case there is no path from any node in $S$ to some node $v \in V \setminus S$, i.e., $\sum_{v \in V \setminus S} \min_{w \in S} d(w,v)=\infty$, we assume that $\cclo(G,S)=0$.
\item \textbf{Betweenness group centrality}~\cite{everett1999centrality} is proportional to percentage of the shortest paths between pairs of non-group nodes that pass through the group:
\[
\cbet(G,S) = \frac{2 \sum_{s \neq t \in V \setminus S} \frac{\pi_{s,t}(S)}{\pi_{s,t}}}{(|V|-|S|)(|V|-|S|-1)},
\]
where $\pi_{s,t}$ is the number of shortest paths between $s$ and $t$, and $\pi_{s,t}(S)$ is the number of paths in $\pi_{s,t}$ that include at least one node from $S$. In case there are no paths from $s$ to $t$, we assume that $\frac{\pi_{s,t}(S)}{\pi_{s,t}} = 0$. The value of the betweenness centrality can be computed efficiently using the algorithm by Brandes~\cite{brandes2008variants}.
\item \textbf{GED-walk group centrality}~\cite{Angriman2020Group} is proportional to the number of random walks in the network that include at least one node from the group:
\[
\cged(G,S) = \sum_{i=1}^\infty \alpha^i \phi_i(S),
\]
where $\alpha > 0$ is the decay parameter, and $\phi_i(S)$ is the number of walks of length $i$ that include at least one node from $S$. The above centrality formula includes an infinite sum. To ensure that the value of the sum converges, in this work we use $\alpha=\frac{1}{\delta^*_G}$, where $\delta^*_G$ is the maximum degree in network $G$. In practical applications, the centrality value can be approximated using an algorithm by Angriman et al.~\cite{Angriman2020Group}. The GED-walk group centrality is inspired by the single-node Katz centrality~\cite{katz1953new}.
\end{itemize}

Let us comment on the selection of group centrality measures considered in this work. Degree, closeness, and betweenness are the most fundamental measures, the latter two being based on the concept of shortest paths, introduced in the work that defined the concept of group centrality~\cite{everett1999centrality}. On the other hand, the GED-walk measure was introduced much more recently, and it is an example of a feedback-based centrality~\cite{das2018study}. Altogether, the considered group centrality measures represent a variety of approaches to the problem of quantifying the importance of a group of nodes.

\section{Problem Statement and Theoretical Analysis}
\label{sec:problem-statement}

We now define the decision and the optimization version of the problem faced by a group of evaders hiding from group centrality measures.

\begin{definition}[Group Hiding]
\label{def:dec-problem}
An instance of this problem is defined by a tuple $(G,\HH,c,\thr,\FR,b)$, where $G=(V,E)$ is a network, $\HH \subseteq V$ is the group of evaders, $c$ is a group centrality measure, $\thr \in \R$ is the safety threshold, $\FR \subseteq \HH \times (\HH \cup N(\HH))$ is the set of edges allowed to be removed, and $b \in \N$ is the hiding budget.  The goal is to identify $\RR \subseteq \FR$ such that $c(G \setminus \RR, \HH) \leq \thr$ and $|\RR| \leq b$.
\end{definition}

\begin{definition}[Minimum Group Hiding]
\label{def:apx-problem}
An instance of this problem is defined by a tuple $(G,\HH,c,\thr,\FR)$, where $G=(V,E)$ is a network, $\HH \subseteq V$ is the group of evaders, $c$ is a group centrality measure, $\thr \in \R$ is the safety threshold, and $\FR \subseteq \HH \times (\HH \cup N(\HH))$ is the set of edges allowed to be removed.  The goal is to identify $R \subseteq \FR$ such that $c(G \setminus \RR, \HH) \leq \thr$ and the size of $R$ is minimal.
\end{definition}

As can be seen from the definitions, the goal in the decision version of the problem is to achieve a predefined level of safety (expressed by the threshold $\thr$) within certain budget $b$ (expressed as the number of removed edges). In the optimization version of the problem the goal is to achieve a predefined level of safety (again, expressed by $\thr$) while removing as few edges from the network as possible. A definition involving the safety threshold is commonly used in the literature on hiding from network measures producing a numerical value~\cite{dey2019covert,waniek2020hiding,waniek2021members,waniek2022social}, often in the context of ranking of all nodes. In this case, the ranking of all groups of nodes would contain an exponential number of positions, hence we focus on satisfying a threshold on the value of the centrality measure.

Let us now comment on the choice of actions available to the evaders. We assume that the evaders can only modify edges that are incident with at least one member of the group, hence $\FR \subseteq \HH \times (\HH \cup N(\HH))$. This assumption is motivated by the fact that in real-life social networks it is extremely difficult to freely affect connections that we are not part of. We only allow the evaders to remove, and not add, edges, as the addition of edges incident with the members of the group can only increase their group centrality for considered measures. In particular, notice that the addition of such an edge can only make a new node a neighbor of the group (thereby increasing group degree centrality), can only decrease the distance to non-group nodes (thereby increasing group closeness centrality), and can only create a new shortest path or a new random walk including a member of the group (thereby increasing group betweenness and GED-walk centralities). Finally, note that, with our problem formulation, it is possible to forbid the removal of some of the edges that are incident with the evaders (hence $\FR \subseteq \HH \times (\HH \cup N(\HH))$ and not $\FR = \HH \times (\HH \cup N(\HH))$), as some connections in social networks might be either impossible to sever (e.g., blood ties), or crucial to the activities of the network.

\begin{table}[t]
\centering
\begin{tabular}{l c c}
\toprule
Centrality & Group Hiding & Minimum Group Hiding \\
\midrule
Degree & \multicolumn{2}{c}{P via Algorithm~\ref{alg:degree-opt} } \\
Closeness & NP-complete & ---\\
Betweenness & NP-complete & APX-hard \\
\bottomrule
\end{tabular}
\caption{The summary of our theoretical findings.}
\label{tab:theoretical-findings}
\end{table}

We now move to analyzing the computational complexity of the problems faced by the group of evaders. In this endeavor, we focus on the centrality measures that provide a concise closed-form formula, i.e., the degree, closeness, and betweenness measures, as they are more amenable to theoretical analysis. Table~\ref{tab:theoretical-findings} summarizes our theoretical findings. The GED-walk measure, which in practice is computed using an approximation algorithm, will be examined more closely in the empirical section. As can be seen, only the simplest group centrality measure, i.e., the degree centrality, allows to find the optimal solution to both versions of the problem in polynomial time. The decision version of the problem is NP-complete for both closeness and betweenness centrality measures, while the optimization version of the problem is also APX-hard given the betweenness centrality. These findings indicate that developing an efficient optimal strategy of hiding a group of nodes from the closeness betweenness centrality measures is most probably impossible.

\begin{algorithm}[tb]
\caption{Optimally decreasing degree group centrality}
\label{alg:degree-opt}
\textbf{Input}: network $G=(V,E)$, group of evaders $\HH \subseteq V$, edges allowed to be removed $\FR \subseteq \HH \times (\HH \cup N(\HH))$, the hiding budget $b \in \N$.
\\
\textbf{Output}: $\RR \subseteq \FR$ the removal of which decreases group degree of $\HH$ optimally within budget
\begin{algorithmic}[1]
\For{$v \in N(\HH)$} \label{ln:degree-opt-1}
\State $X_v = (\{v\} \times \HH) \cap E$ \label{ln:degree-opt-2}
\EndFor
\State $v^* = \langle v \in N(\HH) : X_v \subseteq \FR \rangle$ sorted by non-decreasing $|X_v|$ \label{ln:degree-opt-3}
\State $\RR = \emptyset$ \label{ln:degree-opt-4}
\State $i = 1$ \label{ln:degree-opt-5}
\While{$|\RR| + |X_{v^*_i}| \leq b$} \label{ln:degree-opt-6}
\State $\RR = \RR \cup X_{v^*_i}$ \label{ln:degree-opt-7}
\State $i = i + 1$ \label{ln:degree-opt-8}
\EndWhile
\State \textbf{return} $\RR$ \label{ln:degree-opt-9}
\end{algorithmic}
\end{algorithm}

\begin{theorem}
\label{thr:degree-opt}
Algorithm~\ref{alg:degree-opt} finds $\RR$ the removal of which results in maximal decrease in the degree group centrality of $\HH$ within the budget $b$, i.e.:
\[
\cdeg(G \setminus \RR, \HH) = \min_{R \subseteq \FR : |R| \leq b} \cdeg(G \setminus R, \HH).
\]
Algorithm~\ref{alg:degree-opt} can be used to solve Group Hiding and Minimum Group Hiding problems given the degree group centrality in polynomial time.
\end{theorem}

\begin{proof}
First, let us comment on using Algorithm~\ref{alg:degree-opt} to solve both problems. Let $(G,\HH,c,\thr,\FR,b)$ be an instance of the Group Hiding problem, and let $\RR$ be the output of Algorithm~\ref{alg:degree-opt} given input $G,\HH,\FR,b$. The given instance can be solved by checking whether $\cdeg(G \setminus \RR, \HH) \leq \thr$. If this is the case, then $\RR$ is a solution. Otherwise, the solution does not exist, as $\RR$ provides the maximum decrease within the budget $b$. Similarly, let $(G,\HH,c,\thr,\FR)$ be an instance of the Minimum Group Hiding problem, and let $\RR_b$ be the output of Algorithm~\ref{alg:degree-opt} given input $G,\HH,\FR,b$. The given instance by running Algorithm~\ref{alg:degree-opt} for $b \in \{0, \ldots, |\FR|\}$, and taking $\RR_{b^*}$ where $b^*$ is the smallest $b \leq |\FR|$ such that $\cdeg(G \setminus \RR_b, \HH) \leq \thr$. If no such $b^*$ exists, then the given instance does not have a solution.

We will show that Algorithm~\ref{alg:degree-opt} finds an optimal way to decrease the degree group centrality of $\HH$ within the budget $b$. Notice that the only way to decrease the value of degree group centrality of $\HH$ is to decrease the value of $N(\HH)$. The only way to do it is to completely disconnect a node $v$ in $N(\HH)$ from all nodes in $\HH$ by removing $X_v$ computed in lines~\ref{ln:degree-opt-2}--\ref{ln:degree-opt-3}. Hence, the goal is to accomplish this feat for as many nodes in $N(\HH)$ as possible.

A given node $v \in N(\HH)$ can be disconnected from $\HH$ if and only if all edges connecting $v$ with the nodes in $\HH$ belong to $\FR$. We compute the list of these exact nodes in line~\ref{ln:degree-opt-3}.

As for the proof that the algorithm is optimal, assume that there exists a better solution, i.e., a set $R'$ the removal of which disconnects more nodes in $N(\HH)$ from $\HH$ than the removal of $\RR$. Since $\RR$ is selected based on non-decreasing sizes of $X_v$, the sum of sizes of $X_v$ in $R'$ has to be greater than that in $\RR$. However, the construction of $\RR$ is continued until the sum of sizes in $X_v$ exceeds $b$. Hence, the size of $R'$ has to be greater than the budget $b$.

The time complexity of Algorithm~\ref{alg:degree-opt} is $\mathcal{O}(|E|+|V|\log|V|)$, as computing the values of $X_v$ in lines~\ref{ln:degree-opt-1}--\ref{ln:degree-opt-2} can be done in $\mathcal{O}(|E|)$ steps, sorting in line~\ref{ln:degree-opt-3} can be done in $\mathcal{O}(|V|\log|V|)$ steps, while the loop in lines~\ref{ln:degree-opt-6}--\ref{ln:degree-opt-8} is executed at most $\mathcal{O}(|V|)$ times, with each execution completed in constant time.
\end{proof}

\begin{theorem}
\label{thr:closeness-dec}
The Group Hiding problem is NP-complete given the closeness group centrality.
\end{theorem}

\begin{proof}
The problem is trivially in NP, since after the removal of a given $\RR \in \FR$ we can compute the closeness group centrality of $\HH$ in polynomial time to confirm it is at most $\thr$.

We will prove the NP-hardness of the Group Hiding problem by showing a reduction from the Finding $k$-Clique problem. An instance $(G_C,k)$ of this problem consists of a network $G_C=(V_C,E_C)$, and a constant $k \in \N$. The goal is to determine whether there exist $k$ nodes that induce a clique in $G_C$. The Finding $k$-Clique problem is one the Karp's classic NP-complete problems~\cite{karp1972reducibility}.

\begin{figure}[t]
\centering
\includegraphics[width=.7\columnwidth]{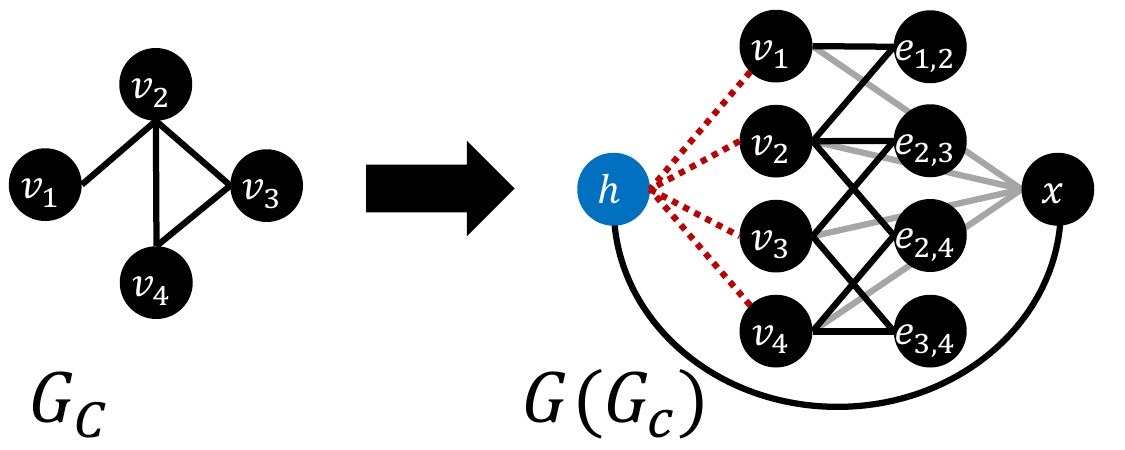}
\caption{Construction used in the proof of Theorem~\ref{thr:closeness-dec}. The blue node belongs to $\HH$, and red dotted edges belong to $\FR$ in the Group Hiding problem instance. Edges between $x$ and nodes $v_i$ are marked grey for better readability.}
\label{fig:closeness-dec}
\end{figure}

Let $X=(G_C,k)$ be an arbitrary instance of the Finding $k$-Clique problem where $V_C=\{v_1,\ldots,v_n\}$. We first construct a network $G(G_C)=(V,E)$ (an example of this construction is presented in Figure~\ref{fig:closeness-dec}):
\begin{itemize}
\item $V = V_C \cup \{h,x\} \cup \bigcup_{(v_i,v_j) \in E_C} \{e_{i,j}\}$,
\item $E = \{(h,x)\} \cup \{(h,v_i), (v_i,x) : v_i \in V_C \} \cup \{(e_{i,j},v_i),(e_{i,j},v_j) : e_{i,j} \in V\}$.
\end{itemize}

We now construct an instance $(G(G_C),\HH,c,\thr,\FR,b)$ of the Group Hiding problem where:
\begin{itemize}
\item $G(G_C)=(V,E)$ is the network we just constructed,
\item $\HH=\{h\}$ is the set of evaders,
\item $c=\cclo$ is the closeness group centrality,
\item $\thr=\frac{|V|-1}{|V_C|+2|E_C|+1 +k+\frac{k(k-1)}{2}}$ is the safety threshold,
\item $\FR=\{(h,v_i) : v_i \in V_C \}$ is the set of edges allowed to be removed,
\item $b=k$ is the hiding budget.
\end{itemize}

In what follows, we will show that the constructed instance of the Group Hiding problem has a solution if and only if the given instance of the Finding $k$-Clique problem has a solution. First, we will prove a useful lemma.

\setcounter{lemma}{1}
\begin{lemma}
\label{lem:construct-closeness}
Let $V' \subseteq V$, and let $R = \{h\} \times V'$.
The set $V'$ induces at least $m$ edges in $G_C$ if and only if $\cclo(G(G_C) \setminus R, \{h\}) \leq \frac{|V_C|-1}{|V_C|+2|E_C|+1+|R|+m}$.
\end{lemma}

\begin{proof}
In what follows, let $G$ denote $G(G_C)$. Let us first compute the distances between $h$ and other nodes in $G$ after the removal of a given $R \subseteq \FR$:
\begin{itemize}
\item $d_{G \setminus R}(h,x) = 1$,
\item $d_{G \setminus R}(h,v_i) = 1$ if $(h,v_i) \notin R$,
\item $d_{G \setminus R}(h,v_i) = 2$ if $(h,v_i) \in R$,
\item $d_{G \setminus R}(h,e_{i,j}) = 2$ if $(h,v_i) \notin R \lor (h,v_j) \notin R$,
\item $d_{G \setminus R}(h,e_{i,j}) = 3$ if $(h,v_i) \in R \land (h,v_j) \in R$. 
\end{itemize}

Hence, we have that for a given $R \subseteq \FR$:
\[
\cclo(G \setminus R, \{h\}) = \frac{|V|-1}{|V_C|+2|E_C|+1+|R|+\alpha},
\]
where $\alpha$ is the number of $e_{i,j}$ such that $(h,v_i) \in R \land (h,v_j) \in R$.

Assume that set $V'$ induces at least $m$ edges in $G_C$. Each of these edges have both ends in $V'$, hence, because of the way in which $G$ is constructed, for nodes $e_{i,j}$ corresponding to these edges we have $(h,v_i) \in R \land (h,v_j) \in R$, which implies $\alpha \geq m$. From the above formula we get that:
\[
\cclo(G \setminus R, \{h\}) \leq \frac{|V|-1}{|V_C|+2|E_C|+1+|R|+m}.
\]
Hence, the first implication of the lemma is correct.

To prove the implication in the other direction, assume that 
\[
\cclo(G(G_C) \setminus R, \{h\}) \leq \frac{|V_C|-1}{|V_C|+2|E_C|+1+|R|+m}.
\]
From the above formula we get that this is the case if and only if $\alpha \geq m$. Hence, there are at least $m$ nodes $e_{i,j}$ such that $(h,v_i) \in R \land (h,v_j) \in R$. Because of the way in which $G$ is constructed, these nodes correspond to edges in $G_C$ with both ends in $V'$. Hence, $V'$ induces at least $m$ edges in $G_C$. This concludes the proof.
\end{proof}

Assume that the constructed instance of the Group Hiding problem has a solution, i.e., a set $\RR$ of size $k$ the removal of which from $G(G_C)$ brings the closeness group centrality of $\HH$ to at most $\thr=\frac{|V|-1}{|V_C|+2|E_C|+1 +k+\frac{k(k-1)}{2}}$. Let $V^*=\{v_i \in V_C : (h,v_i) \in \RR\}$. From Lemma~\ref{lem:construct-closeness} we get that $V^*$ induces $m$ edges in $G_C$, where $|\RR|+m=k+\frac{k(k-1)}{2}$. Since $|\RR|=k$, we get that $m=\frac{k(k-1)}{2}$, i.e., $V^*$ induces $\frac{k(k-1)}{2}$ edges in $G_C$. However, since $|V^*|=k$, it is only possible if $V^*$ induces a clique. Hence, $V^*$ is a solution to the given instance of the Finding $k$-Clique problem.

Assume that the given instance of the Finding $k$-Clique problem has a solution, i.e., a set of $k$ nodes $V^*$ that induces a clique in $G_C$. Since $k$-clique has $\frac{k(k-1)}{2}$ edges, from Lemma~\ref{lem:construct-closeness} we get that for $\RR=\{h\} \times V^*$ we have $\cclo(G(G_C) \setminus \RR, \{h\}) \leq \frac{|V|-1}{|V_C|+2|E_C|+1+k+\frac{k(k-1)}{2}}=\thr$. Since $|\RR|=k$, we have that $\RR$ is a solution to the constructed instance of the Group Hiding problem.

We showed that the constructed instance of the Group Hiding problem has a solution if and only if the given instance of the Finding $k$-Clique problem has a solution, thus completing the proof of NP-hardness.
\end{proof}

\begin{theorem}
\label{thr:betweenness-dec}
The Group Hiding problem is NP-complete given the betweenness group centrality.
\end{theorem}

\begin{proof}
The problem is trivially in NP, since after the removal of a given $\RR \in \FR$ we can compute the betweenness group centrality of $\HH$ in polynomial time to confirm it is at most $\thr$.

We will prove the NP-hardness of the Group Hiding problem by showing a reduction from the Multiway Cut problem. An instance $(G_M,S,\mu)$ of this problem consists of a network $G_M=(V_M,E_M)$, a set of terminal nodes $S \subseteq V_M$, and a constant $\mu \in \N$. The goal is to determine whether there exist $\mu$ edges in $E_M$ such that after their removal there exists no path between any pair of nodes in $S$. The Multiway Cut problem is NP-hard given that $|S| \geq 3$~\cite{dahlhaus1994complexity}.

\begin{figure}[t]
\centering
\includegraphics[width=.7\columnwidth]{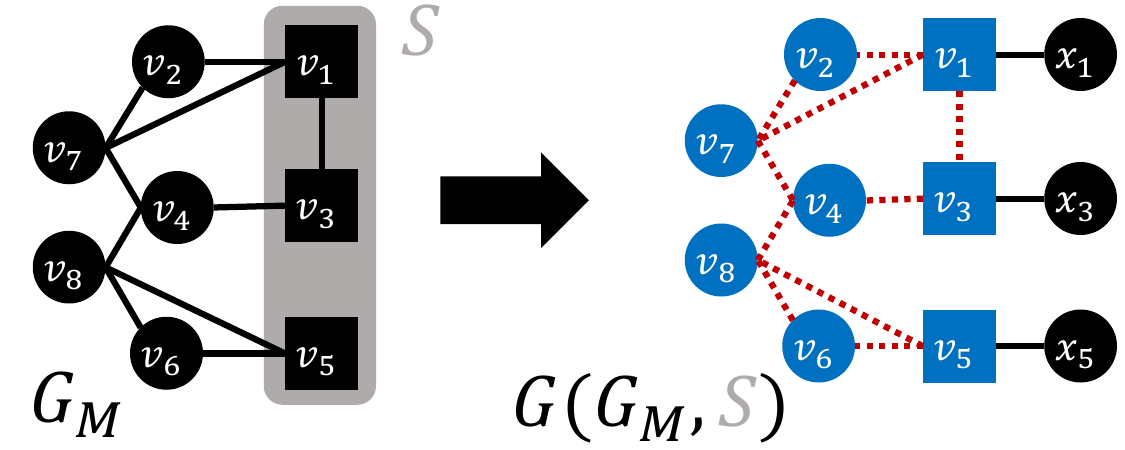}
\caption{Construction used in the proof of Theorem~\ref{thr:betweenness-dec}. Square nodes belong to $S$ in the Multiway Cut problem instance. Blue nodes belong to $\HH$, and red dotted edges belong to $\FR$ in the Group Hiding problem instance.}
\label{fig:betweenness-dec}
\end{figure}

Let $X=(G_M,S,\mu)$ be an arbitrary instance of the Multiway Cut problem where $|S| \geq 3$, and where $V_M=\{v_1,\ldots,v_n\}$. We first construct a network $G(G_M,S)=(V,E)$ (an example of this construction is presented in Figure~\ref{fig:betweenness-dec}):
\begin{itemize}
\item $V = V_M \cup \bigcup_{v_i \in S} \{x_i\}$,
\item $E = E_M \cup \bigcup_{v_i \in S}\{(v_i,x_i)\}$.
\end{itemize}

We now construct an instance $(G(G_M,S),\HH,c,\thr,\FR,b)$ of the Group Hiding problem where:
\begin{itemize}
\item $G(G_M,S)=(V,E)$ is the network we just constructed,
\item $\HH=V_M$ is the set of evaders,
\item $c=\cbet$ is the betweenness group centrality,
\item $\thr=0$ is the safety threshold,
\item $\FR=E_M$ is the set of edges allowed to be removed,
\item $b=\mu$ is the hiding budget.
\end{itemize}

In what follows, we will show that the constructed instance of the Group Hiding problem has a solution if and only if the given instance of the Multiway Cut problem has a solution. First, we will prove a useful lemma.

\begin{lemma}
\label{lem:construct-betweenness}
The removal of a set $R \subseteq E_M$ from $G_M=(V_M,E_M)$ disconnects all pairs of nodes in $S \subseteq V_M$ if and only if $\cbet(G(G_M,S) \setminus R, V_M) \leq 0$.
\end{lemma}

\begin{proof}
In what follows, let $G$ denote $G(G_M,S)$.

Assume that the removal of $R$ from $G_M$ disconnects all pairs of nodes in $S$.  From the definition of the betweenness group centrality we have that:
\[
\cbet(G \setminus R,V_M) = \frac{2}{|S|(|S|-1)} \sum_{x_i \neq x_j} \frac{\pi_{x_i,x_j}(V_M)}{\pi_{x_i,x_j}}.
\]
Since the only node connected with any $x_i \in V$ is $v_i$, any path between $x_i$ and $x_j$ have to include the path between $v_i$ and $v_j$ running through the nodes in $V_M$. However, since $v_i,v_j \in S$ and since the only edges between the nodes in $V_M$ are the edges in $E_M$, the removal of $R$ from $G$ disconnects any such pair of nodes $v_i$ and $v_j$. Consequently, there is no path between any pair of nodes $x_i$ and $x_j$ in $G \setminus R$, hence $\cbet(G \setminus R, V_M) = 0$.

To prove the implication in the other direction, assume that $\cbet(G \setminus R, V_M) \leq 0$. From the definition of the betweenness group centrality we have that if there exists at least one path between nodes $x_i$ and $x_j$ in $G \setminus R$ then $\cbet(G \setminus R, V_M) > 0$, as any such path must include at least the nodes $v_i,v_j \in V_M$, the only nodes connected to $x_i$ and $x_j$. Therefore, the removal of $R$ from $G$ disconnects all pairs of nodes $v_i,v_j$. Because of the way in which the network $G$ is constructed, these are precisely all the pairs of nodes in $S$. Since the only edges between the nodes in $V_M$ are the edges in $E_M$, the removal of $R$ from $G_M$ disconnects all pairs of nodes in $S$. This concludes the proof.
\end{proof}

Assume that the constructed instance of the Group Hiding problem has a solution, i.e., a set $\RR$ of size at most $b$ the removal of which from $G(G_M,S)$ brings the betweenness group centrality of $\HH$ to at least $\thr$. Since $\HH=V_M$ and $\thr = 0$, we have that $\cbet(G(G_M,S) \setminus R, V_M) \leq 0$. Therefore, from Lemma~\ref{lem:construct-betweenness} we get that the removal of $\RR$ from $G_M$ disconnects all pairs of nodes in $S$. Since $|\RR| \leq b = \mu$, the set $\RR$ is also a solution to the given instance of the Multiway Cut problem.

Assume that the given instance of the Multiway Cut problem has a solution, i.e., a set $R$ of size at most $\mu$ the removal of which from $G_M$ disconnects all pairs of nodes in $S$. From Lemma~\ref{lem:construct-betweenness} we get that  $\cbet(G(G_M,S) \setminus R, V_M) \leq 0$. Since $\HH=V_M$, $\thr = 0$, and $b=\mu$, $R$ is also a solution to the constructed instance of the Group Hiding problem.

We showed that the constructed instance of the Group Hiding problem has a solution if and only if the given instance of the Multiway Cut problem has a solution, thus completing the proof of NP-hardness.
\end{proof}

\begin{theorem}
\label{thr:betweenness-apx}
The Minimum Group Hiding problem is APX-hard given the betweenness group centrality.
\end{theorem}

\begin{proof}
We will prove the APX-hardness of the Minimum Group Hiding problem by showing an L-reduction from the Minimum Multiway Cut problem. An instance $(G_M,S)$ of this problem consists of a network $G_M=(V_M,E_M)$, and a set of terminal nodes $S \subseteq V_M$. The goal is to find $F \subseteq E_M$ such that the removal of $F$ from $G_M$ disconnects all pairs of nodes in $S$ and the size of $R$ is minimal. The Minimum Multiway Cut problem is APX-hard given that $|S| \geq 3$~\cite{dahlhaus1994complexity}.

The L-reduction consists of a pair of polynomial time computable functions $p$ and $q$ such that for  an arbitrary instance $X=(G_M,S)$ of the Minimum Multiway Cut problem:

\begin{enumerate}
\item $p(X)$ is an instance of the Minimum Group Hiding problem;
\label{pt:apx-hard-1}

\item if $R$ is a solution to $p(X)$ then $q(R)$ is a solution to $X$;
\label{pt:apx-hard-2}

\item An optimal solution $F^*$ to $X$ and an optimal solution $\RR$ to $p(X)$ are of the same size, i.e., $|\RR| = |F^*|$;
\label{pt:apx-hard-3}

\item For every solution $R$ to $p(X)$, we have that $\left| |F^*| - |q(R)| \right| = \left| |\RR| - |R| \right|$, where $F^*$ is an optimal solution to $X$ and $\RR$ is an optimal solution to $p(X)$.
\label{pt:apx-hard-4}
\end{enumerate}

We first define the function $p$ (point~\ref{pt:apx-hard-1} of the L-reduction definition) as:
\[
p(G_M,S) = (G(G_M,S),\HH,c,\thr,\FR),
\]
where:
\begin{itemize}
\item $G(G_M,S)$ is the network defined in the proof of Theorem~\ref{thr:betweenness-dec},
\item $\HH=V_M$ is the set of evaders,
\item $c=\cbet$ is the betweenness group centrality,
\item $\thr=0$ is the safety threshold,
\item $\FR=E_M$ is the set of edges allowed to be removed.
\end{itemize}

Second, we define the function $q$ as $q(R)=R$, i.e., the solution to the instance $X$ of the Minimum Multiway Cut problem is the same set of edges as the solution to the instance $p(X)$ of the Minimum Group Hiding problem. We will now show that if $R$ is a solution to $p(X)$, then $q(R)$ is indeed a solution to $X$ (point~\ref{pt:apx-hard-2} of the L-reduction definition). Indeed from Lemma~\ref{lem:construct-betweenness} we get that if $\cbet(G(G_M,S) \setminus R, V_M) \leq 0$ (which is a necessary condition for $R$ to be a solution to $p(X)$, based on the definition of the Minimum Group Hiding problem), then the removal of $R$ from $G_M$ disconnects all pairs of nodes in $S$. Hence, $q(R)=R$ is a solution to $X$.

To prove point~\ref{pt:apx-hard-3} of the L-reduction definition, assume that $X$ and $p(X)$ have optimal solutions of different sizes. However, from Lemma~\ref{lem:construct-betweenness} we get that a solution to one of the problem instances is also a solution to the other instance. Hence, the smaller of the two optimal solution is a solution to the other instance, smaller than the optimum, which causes a contradiction. Therefore, both instances necessarily have the same optimal solution size.

Finally, notice that since $|F^*|=|\RR|$ (as argued in the previous paragraph) and $q(R)=R$ (from the definition of function $q$ above) we have:
\[
\left| |F^*| - |q(R)| \right| = \left| |\RR| - |R| \right|.
\]
Therefore, point~\ref{pt:apx-hard-4} of the L-reduction definition is satisfied.

Therefore, we showed that $p$ and $q$ constitute an L-reduction from the Minimum Multiway Cut problem to the problem of Minimum Group Hiding problem given the betweenness group centrality. Since the Minimum Multiway Cut problem is APX-hard~\cite{dahlhaus1994complexity}, the Minimum Group Hiding problem is also APX-hard, as L-reductions preserve APX-hardness~\cite{crescenzi1997short}.
\end{proof}

\section{Experimental Analysis}
\label{sec:experimental-analysis}

In this section we conduct experiments comparing the effectiveness of different ways of hiding from group centralities in randomly-generated and real-life networks.

\subsection{Network Datasets}

In our simulations we use the following models of generating random networks:
\begin{itemize}
\item \textbf{\BAn model}~\cite{barabasi1999emergence} generating scale-free networks with preferential attachment,
\item \textbf{\WSn model}~\cite{watts1998collective} generating small-world networks,
\item \textbf{\ERn model}~\cite{erdds1959random} generating networks with a uniform distribution of edges.
\end{itemize}

Unless stated otherwise, in our simulation we consider random networks with $5000$ nodes, and an average degree of $10$. In case of the \WSn model we set the rewiring probability to $\frac{1}{4}$.

\begin{table}[t!]
\centering
\begin{tabular}{l c c}
\toprule
Network & $|V|$ & $|E|$ \\
\midrule
Bitcoin Alpha & $3,775$ & $14,120$ \\
Bitcoin OTC & $5,875$ & $21,489$ \\
Gnutella & $6,299$ & $20,776$ \\
PGP & $10,680$ & $24,316$ \\
\bottomrule
\end{tabular}
\caption{The properties of real-life network datasets.}
\label{tab:real-life}
\end{table}

We also consider the following real-life networks (their details are presented in Table~\ref{tab:real-life}):
\begin{itemize}
\item \textbf{Bitcoin Alpha}~\cite{kumar2016edge} the giant component of the trust network between the users of Bitcoin Alpha trading platform,
\item \textbf{Bitcoin OTC}~\cite{kumar2016edge} the giant component of the trust network between the users of Bitcoin OTC trading platform,
\item \textbf{Gnutella}~\cite{ripeanu2002mapping} the network of connections between hosts of Gnutella peer-to-peer sharing network (snapshot from August 8, 2002),
\item \textbf{PGP}~\cite{boguna2004models} the network of users who signed their public keys using Pretty-Good-Privacy algorithm for secure information exchange in 2004.
\end{itemize}

\subsection{Experimental Procedure}

Given a network $G=(V,E)$, we first select the group of evaders $\HH$. We consider the following ways of doing so (all ties are resolved uniformly at random):
\begin{itemize}
\item \textbf{Dense}---first we add a random node to $\HH$, and then we repeatedly add a node that maximizes the number of connections with the members of $\HH$. 
\item \textbf{Cells}---$\HH$ consists of a number of separate cells with sizes selected using an algorithm that generates terrorist cells~\cite{tsvetovat2005generation}, with the members of each cell selected using the same procedure as for the Dense selection criterion. 
\item \textbf{Scattered}---the members of $\HH$ are selected uniformly at random out of all the network nodes. 
\end{itemize}

We chose these particular three ways of selecting $\HH$ to compare across a variety of possible distributions of evaders in the network. The first method results in a cohesive group of evaders, consistent with an assumption common in the community detection literature that closely cooperating groups of nodes typically have many intra-community connections~\cite{fortunato2010community}. The second method results in multiple smaller groups of nodes spread across the network, a structure often exhibited by terrorist and criminal organizations~\cite{tsvetovat2005generation}, who have strong incentives to remain undetected. Finally, the last method results in evaders distributed across the network, and provides a strong contrast to the other two methods.

Unless stated otherwise, we set the size of $\HH$ to $50$ in our experiments. We assume that the evaders can hide by removing from the network edges in $\FR$, i.e., edges incident with at least one member of $\HH$. We consider the following hiding strategies:
\begin{itemize}
\item \textbf{Optimal degree}---removes an edge selected using Algorithm~\ref{alg:degree-opt} (the optimal strategy of hiding from the degree centrality). We consider this strategy as the only optimal hiding strategy (against the considered centrality measures) that can be computed in polynomial time.
\item \textbf{Internal}---removes a random edge from $\FR$ connecting two members of $\HH$. This is probably the easiest strategy to implement in practice, as it does not require any cooperation from nodes outside the group of evaders. In principle, this strategy could be effective against  betweenness and GED-walk centralities, as it may result in breaking shortest paths or random walks including two or more evaders.
\item \textbf{Random}---removes a random edge from $\FR$. This strategy serves as a baseline against which other strategies can be compared.
\item \textbf{Shortcut}---remove an edge between a member of $\HH$ and a node $v^* \in N(\HH)$ such that the average distance between $v^*$ and all other nodes in $G \setminus \HH$ is minimal. This strategy is based on an idea of disconnecting $\HH$ from nodes that provide quick access to the rest of the network, as both the closeness and betweenness group centrality measures operate based on the shortest paths.
\end{itemize} 

For any given combination of the network and the selection criterion, we select $10$ different groups $H$. For each group we attempt to hide it using each of the strategies while setting the hiding budget (i.e., the number of edges that can be removed from the network) to $|\HH|$. We measure the value of each centrality measure before and after the hiding process. In case of the random network generation models, we repeat the procedure $100$ times, each time generating a new network.

\begin{figure*}[t]
\centering
\includegraphics[width=\linewidth]{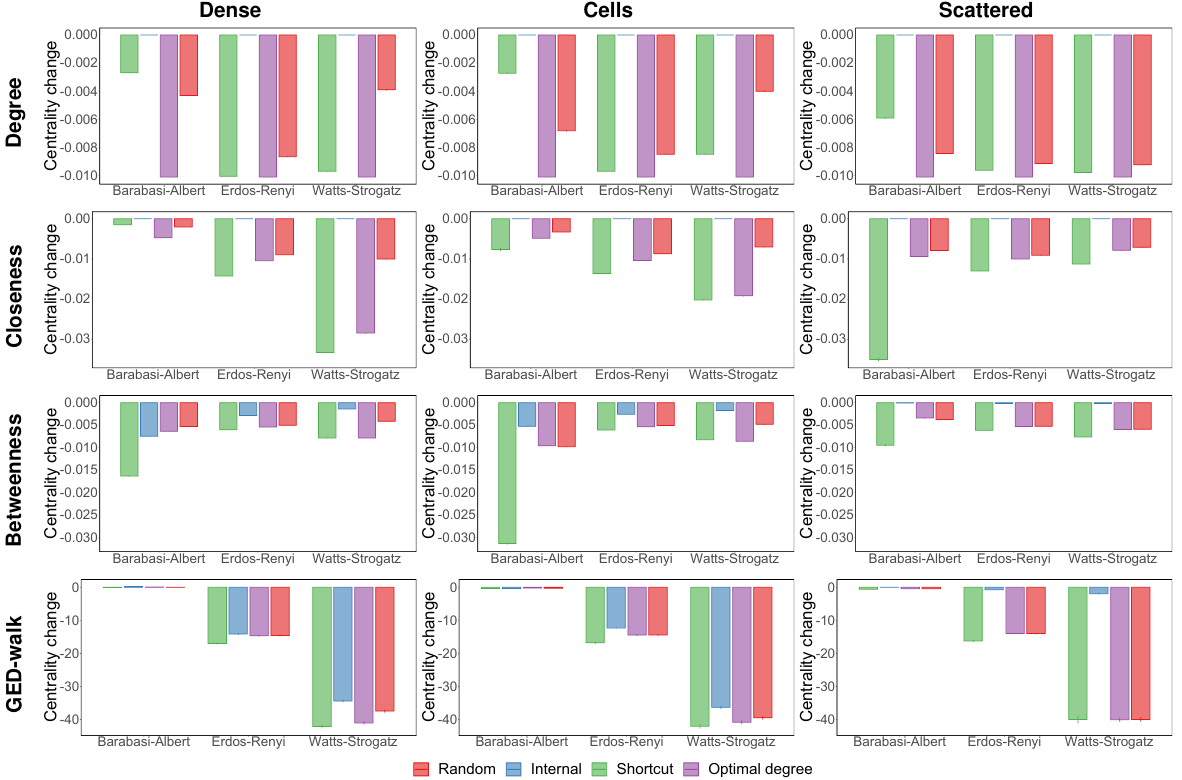}
\caption{Results of hiding groups in random networks with $5000$ nodes and an average degree of $10$. Each row corresponds to a different centrality measure, while each column corresponds to a different way of selecting $\HH$. In each plot the x-axis corresponds to a different network generation model, while the y-axis corresponds to the change in centrality value after the hiding process. Each bar represents a different hiding strategy, with the error bars representing $95\%$-confidence intervals.
}
\label{fig:basic-bars-random}
\end{figure*}

Given a very different scale of possible values for different centrality measures, we do not set a fixed safety threshold $\thr$, but rather present a comparison of what thresholds could be satisfied in different networks. Keep in mind that in practical application, a group of evaders might be unsure which centrality is going to be used against them, hence they might want to hide from a variety of measures, rather than set a specific threshold for a single one.

\begin{figure*}[t]
\centering
\includegraphics[width=\linewidth]{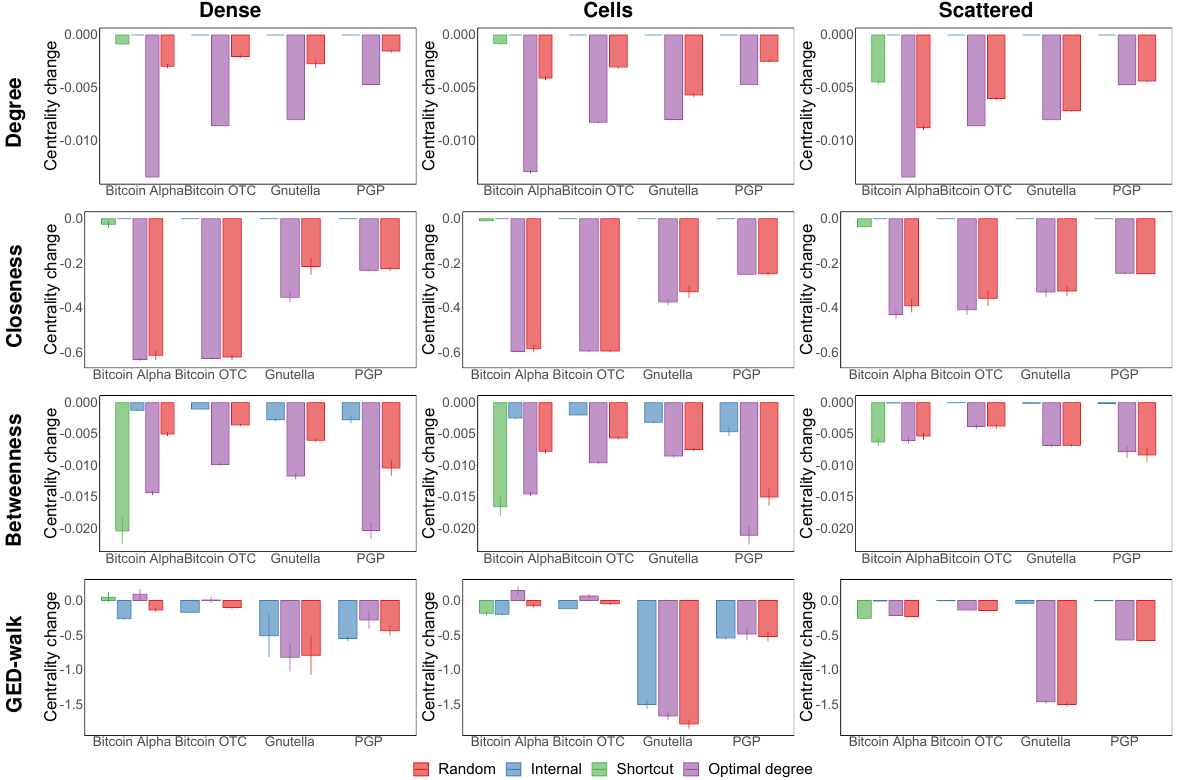}
\caption{Results of hiding groups in real-life networks. Each row corresponds to a different centrality measure, while each column corresponds to a different way of selecting $\HH$. In each plot the x-axis corresponds to different network dataset, while the y-axis corresponds to a change in centrality value after the hiding process. Each bar represents different hiding strategy, with the error bars representing $95\%$-confidence intervals.
}
\label{fig:basic-bars-real}
\end{figure*}

It is worth noting that after removal some of the edges during the hiding process, the network might become disconnected. Nevertheless, it does not pose a problem, as all of the considered centrality measures have well defined values given a disconnected network. The degree centrality is only concerned with immediate neighbors of $\HH$, not matter the connectivity of the network. The closeness centrality of $\HH$ drops to zero if there exist connected components that do not contain any members of $\HH$ (as stated in Section~\ref{sec:preliminaries}). The betweenness centrality set the contributions of pair of nodes in different connected components to zero (again, as stated in Section~\ref{sec:preliminaries}). The GED-walk centrality considers all walks of a given length, which in a disconnected network can be computed separately for each connected component.

\subsection{Results of the Experiments}

\begin{figure*}[t]
\centering
\includegraphics[width=\linewidth]{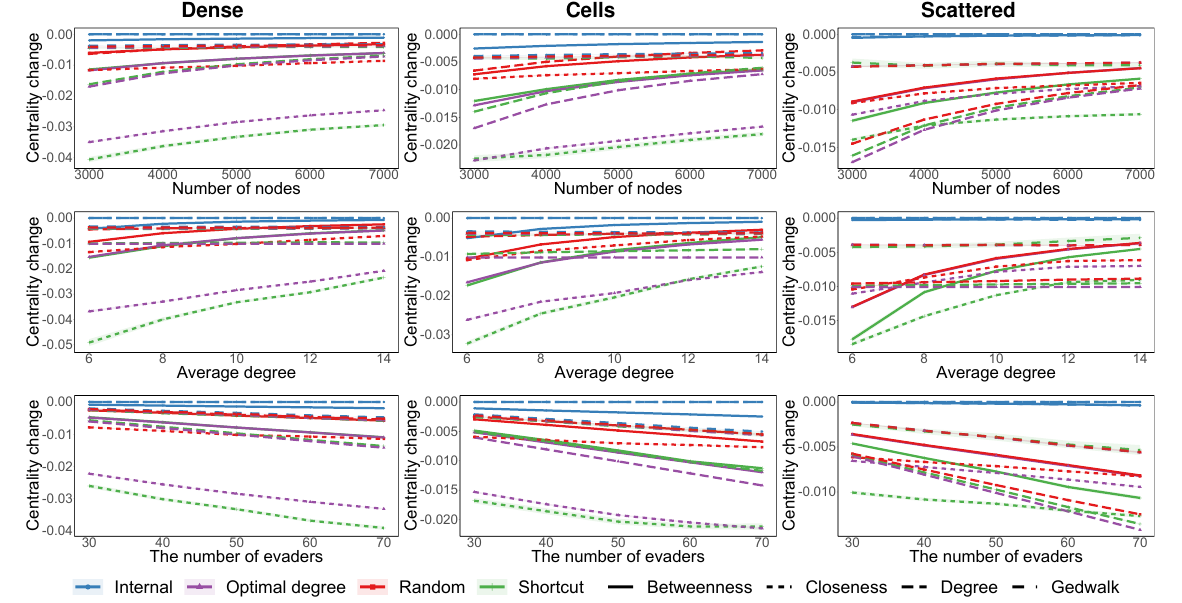}
\caption{Results of hiding groups in \WSn networks with changing parameters. Each row corresponds to a different parameter being modified, either the size of the network, its average degree, or the size of $\HH$. Each column corresponds to a different way of selecting $\HH$. In each plot the x-axis corresponds the selected parameter, while the y-axis corresponds to the change in centrality value after the hiding process. Each line type corresponds to a different centrality measure, while each color and point shape correspond to a different hiding strategy. Colored areas (very narrow in most cases) represent $95\%$-confidence intervals. To improve readability at the same scale, values for GED-walk centrality are multiplied by $10^{-4}$.
}
\label{fig:params-ws}
\end{figure*}

Figure~\ref{fig:basic-bars-random} summarizes the outcomes of the hiding process in random networks, while Figure~\ref{fig:basic-bars-real} presents the results for real-life networks. Regarding the effectiveness of different hiding strategies, the optimal degree strategy proves to be relatively effective against all centrality measures, often exhibiting the best performance out of the three strategies utilizing only local information. The shortcut strategy is also effective in most cases, although it requires information about the entire network, making it less applicable in practice. Removing edges between the evaders is ineffective in most cases, performing worse than the random baseline. Interestingly, groups selected using all three criteria, i.e., dense, cells, and scattered, are equally difficult to hide, despite representing very different types of structures. Likewise, the hiding strategies have similar levels of performance in networks generated using all three models, i.e., \BAn, \WSn, and \ERn. These findings indicate that the most crucial aspect of hiding from group centrality is the centrality measure used to run the analysis, rather the character of the network structure or the distribution of evaders throughout the network.

The results above are presented for networks with $5000$ nodes. In real-life scenarios, the networks can be much larger, consisting of hundreds of thousands or even millions of nodes. However, it is worth noting that the main computational bottleneck of our simulations is the computation of centrality measures, rather than the hiding strategies. The time complexity of most of our hiding heuristics (other than the shortcut heuristic) depends only on the size of the network vicinity of $\HH$, rather than the size of the entire network. In other words, it is possible to employ our hiding heuristics even in much larger networks, in which evaluating their effectiveness is computationally unfeasible (particularly in terms of the betweenness centrality, the computation of which is the most time consuming).

\begin{figure*}[t]
\centering
\includegraphics[width=\linewidth]{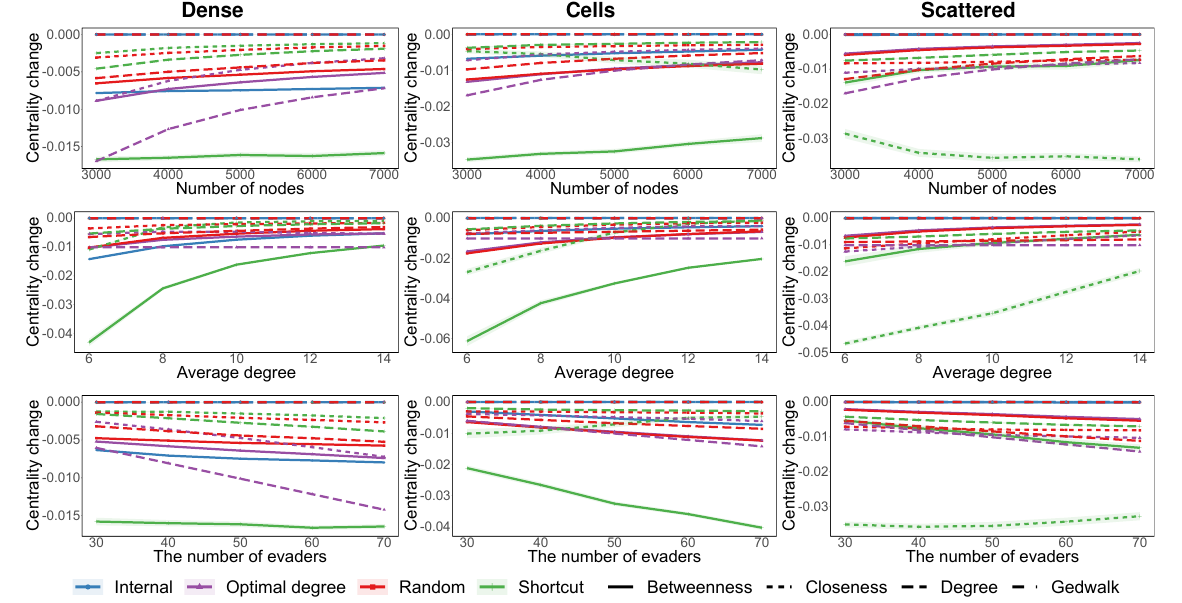}
\caption{Results of hiding groups in \BAn networks with changing parameters. Each row corresponds to a different parameter being modified, either the size of the network, its average degree, or the size of $\HH$. Each column corresponds to a different way of selecting $\HH$. In each plot the x-axis corresponds the selected parameter, while the y-axis corresponds to a change in centrality value after the hiding process. Each line type corresponds to a different centrality measure, while each color and point shape to a different hiding strategy. Colored areas (very narrow in most cases) represent $95\%$-confidence intervals. To improve readability at the same scale, values for GED-walk centrality are multiplied by $10^{-4}$.
}
\label{fig:params-ba}
\end{figure*}

\begin{figure*}[tbh]
\centering
\includegraphics[width=\linewidth]{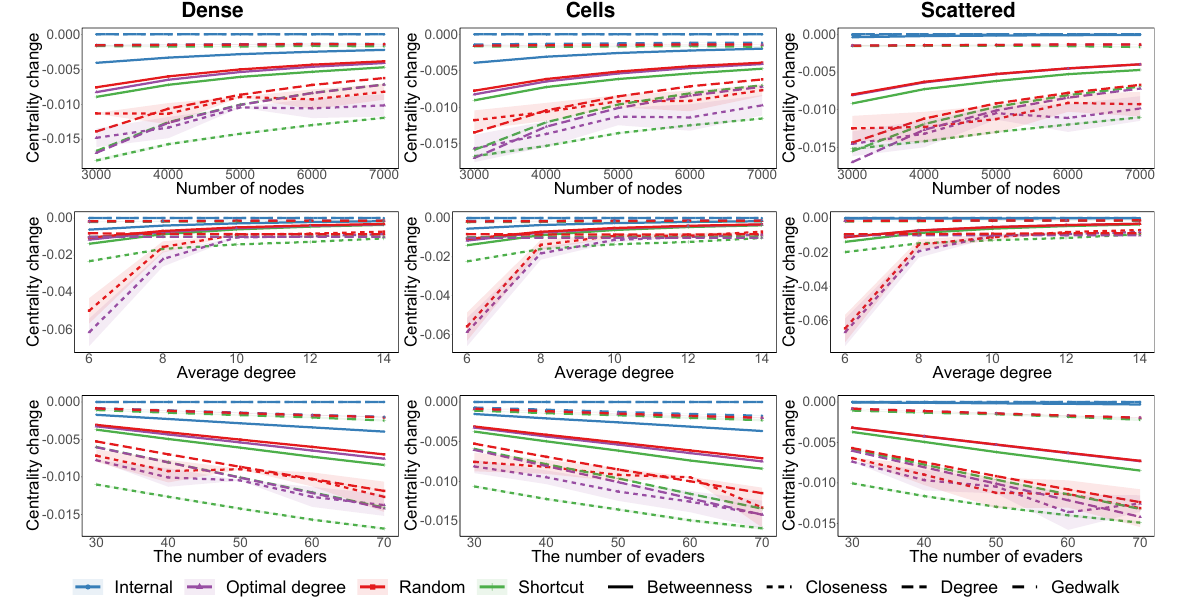}
\caption{Results of hiding groups in \ERn networks with changing parameters. Each row corresponds to a different parameter being modified, either the size of the network, its average degree, or the size of $\HH$. Each column corresponds to a different way of selecting $\HH$. In each plot the x-axis corresponds the selected parameter, while the y-axis corresponds to a change in centrality value after the hiding process. Each line type corresponds to a different centrality measure, while each color and point shape to a different hiding strategy. Colored areas (very narrow in most cases) represent $95\%$-confidence intervals. To improve readability at the same scale, values for GED-walk centrality are multiplied by $10^{-4}$.
}
\label{fig:params-er}
\end{figure*}

Given that the effectiveness of hiding does not seem to depend on the model used to generate the network structure, can the same be said about the size and the density of the network, as well as the size of the group of evaders? Figures~\ref{fig:params-ws}, \ref{fig:params-ba} and \ref{fig:params-er} present the results of our simulations with \WSn, \BAn, and \ERn networks respectively, while varying the number of nodes, the average degree of the nodes, and the size of $\HH$. As can be seen in the figures, the effectiveness of hiding strategies is on average greater in small and sparse networks, and for groups with many members (keep in mind that we assume the hiding budget equal to the size of the group of evaders). These seem like relatively good news for the evaders, as real-life networks are often sparse, and enlisting additional members can actually make obscuring the importance of the entire organization easier.

\begin{figure*}[tbh]
\centering
\includegraphics[width=\linewidth]{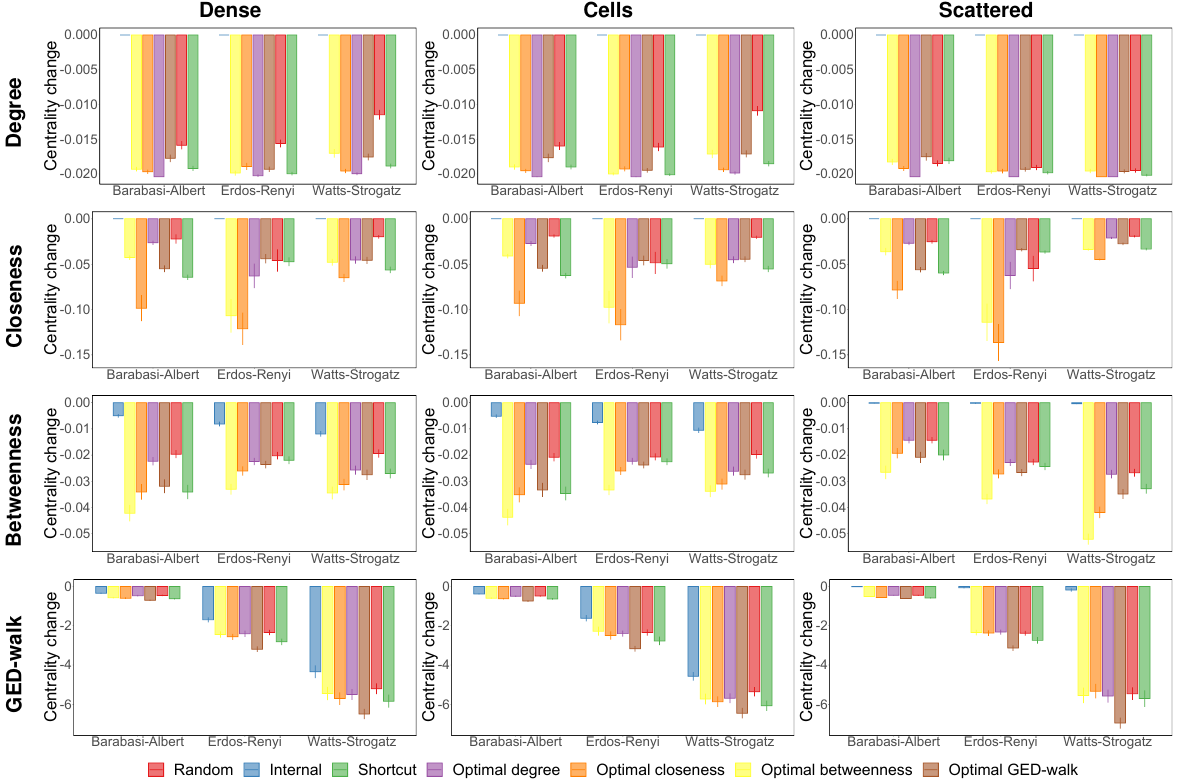}
\caption{Results of hiding groups in random networks when compared to optimal strategies. Each row corresponds to a different centrality measure, while each column corresponds to a different way of selecting $\HH$. In each plot the x-axis corresponds to different network generation model, while the y-axis corresponds to the change in centrality value after the hiding process. Each bar represents a different hiding strategy, with the error bars representing $95\%$-confidence intervals.
}
\label{fig:optimal-bars}
\end{figure*}

So far, the only centrality measure against which we were able to hide optimally was the degree centrality. Now, we compare the effectiveness of the proposed hiding strategies to the optimal hiding strategies for closeness, betweenness, and GED-walk centralities. To this end, we consider networks with $200$ nodes and an average degree of $4$, as well as groups of evaders with $4$ members. In these much smaller networks we are able to examine all possible subsets of $\FR$ of size at most $b$, in order to identify the optimal hiding strategy. The results of this analysis are presented in Figure~\ref{fig:optimal-bars}. As can be seen from the figure, the optimal degree strategy, which can be computed in polynomial time, has an effectiveness that falls between $50\%$ and $100\%$ compared to that of the optimum strategies in most cases; the only exception being the closeness centrality in scale-free networks. This finding suggests that in situations in which the evaders cannot evaluate all possible ways of hiding, or they are not sure which centrality measure is going to be used against them, running the optimal degree strategy seems to be a viable choice.

\section{Conclusions}
\label{sec:conclusions}

In this work we analyzed the problem of hiding the importance of a group of nodes from group centrality measures. We showed that finding an optimal solution in polynomial time is possible for the degree group centrality. However, the same problem is NP-complete for the closeness and the betweenness measures. Moreover, approximating the optimal solution is APX-hard for the betweenness group centrality. Nevertheless, our empirical analysis suggests that the optimal strategy against degree centrality is a viable choice against the other two centrality measures. We also found that it is generally easier to hide in small and sparse networks, even if the size of the group increases. Altogether, our work is the first step in the rigorous analysis of the problem of evading group centrality measures.

Potential directions of future work include analyzing alternative group centrality measures, designing a network structure where a specific group maintains low values of all centrality measures, and developing group centralities that are hard to evade.

\bibliographystyle{abbrv}
\bibliography{bibliography-group-centrality}

\end{document}